\documentclass[runningheads,a4paper]{llncs}

\usepackage{amssymb}
\usepackage{graphicx}

\usepackage{amsmath}
\usepackage{amstext}
\usepackage{amsfonts}
\usepackage{amsxtra}
\usepackage[linesnumbered,boxed,german]{algorithm2e}
\SetAlCapSkip{3pt}

\usepackage[T1]{fontenc}
\usepackage[utf8]{inputenc}

\usepackage{a4wide}

\usepackage{tikz}
\usetikzlibrary{arrows,shapes,patterns}

\usepackage{paralist}

\usepackage{cite} 

\newcommand{\abc}[3]{$#1 \,|\, #2 \,|\, #3$}

\newcommand{\eps}{\ensuremath{\varepsilon}}

\newcommand{\Sch}{\ensuremath{\mathcal{S}}}  

\newtheorem{observation}{Observation}

\usepackage{tabularx,booktabs}
\newcolumntype{Y}{>{\raggedright\arraybackslash}X}
\newcolumntype{Z}{>{\centering\arraybackslash}X}

\usepackage[textsize=tiny]{todonotes}

\newcommand{\reserved}{utilized}
\newcommand{\reserve}{utilize}
\newcommand{\reserves}{utilizes}
\newcommand{\reserving}{utilizing}
\newcommand{\reservation}{utilization}
\newcommand{\tff}{tariff}
\newcommand{\tffs}{tariffs}

\begin{document}

\mainmatter 

\title{Optimal Algorithms for Scheduling under\\ Time-of-Use Tariffs
    \thanks{A preliminary version of this paper with a subset
    of results appeared in the Proceedings of MFCS 2015
    as~\cite{ChenMRSV15}. This research was supported by the German Science
    Foundation (DFG) under contract ME~3825/1.}}


\author{\mbox{Lin Chen\inst{1} \and Nicole Megow\inst{2} \and Roman
    Rischke\inst{3} \and Leen Stougie\inst{4} \and Jos\'{e}
    Verschae\inst{5}}}


\institute{
  Department of Computer Science, University of Houston, USA.  \email{chenlin198662@gmail.com}%
  \and
  Department of Mathematics and Computer
  Science, University of Bremen, Germany. \email{nicole.megow@uni-bremen.de}%
  \and
  Department of Video Coding \& Analytics, Fraunhofer Heinrich Hertz Institute Berlin, Germany. \email{roman.rischke@gmail.com}%
  \and CWI \& Department of Econometrics and Operations Research, Vrije Universiteit Amsterdam \& Erable-INRIA, The
Netherlands. \email{stougie@cwi.nl} \and Institute of Engineering Sciences, Universidad de O'Higgins, Chile. \email{jose.verschae@uoh.cl} }

\maketitle

\begin{abstract}
  We consider a natural generalization of classical scheduling
  problems in which using a time unit for processing a job causes some
  time-dependent cost which must be paid in addition to the standard
  scheduling cost.  We study the scheduling objectives of minimizing
  the makespan and the sum of (weighted) completion times. It is not
  difficult to derive a polynomial-time algorithm for preemptive
  scheduling to minimize the makespan on unrelated machines.  The
  problem of minimizing the total (weighted) completion time is
  considerably harder, even on a single machine.  We present a
  polynomial-time algorithm that computes for any given sequence of
  jobs an optimal schedule, i.e., the optimal set of time-slots to be
  used for scheduling jobs according to the given sequence. This
  result is based on dynamic programming using a subtle analysis of
  the structure of optimal solutions and a potential function
  argument. With this algorithm, we solve the unweighted problem
  optimally in polynomial time.     
  For the more general problem, in which jobs may have individual weights, we develop a polynomial-time approximation scheme~(PTAS) based on a dual scheduling approach introduced for scheduling on a machine of varying speed. As the weighted problem is strongly NP-hard, our PTAS is the best possible approximation we can hope for.
\end{abstract}

\section{Introduction}
\label{sec:intro}

One of the classical operations research problems is the Production
Planning problem. It appears in almost any introductory course in
Operations Research \cite{Hillier,Taha}. In its deterministic form a
production plan at lowest total cost is required to meet known
demands in the next few weeks, given holding cost for keeping
inventory at the end of the week, and with unit production cost
varying over the weeks. It is a very early example of a problem
model in which unit cost, or tariffs, for production, service,
labor, energy, etc., {\em vary over time}.

Nowadays, new technologies allow direct communication of a much
larger variety of time-of-use tariffs to customers.
E.g. in energy practice electricity prices can differ largely over
the hours. Producers or providers of these resources use these
variable pricing more and more to spread demand for their services,
which can save enormously on the excessive costs that are usually
involved to serve high peak demands. Customers are persuaded to
direct their use of the scarce resources to time slots that are
offered at cheaper rates. From the provider's point of view variable
pricing problems have been studied quite extensively. For instance,
revenue management is a well established subfield of operations
research \cite{talluri2006theory}.

As in the Production Planning problem, in this paper we advocate
models from the point of view of the user of the resources, who may
take advantage from variable pricing by traveling, renting labor,
using electricity, etc. at moments at which these services are
offered at a lower price. This point of view forms a rich class
of optimization problems in which next to classical objectives, the
cost of using services needs to be taken into account.

This widely applicable framework is particularly well suited
for scheduling problems, in which jobs need to be scheduled over
time. Processing jobs requires labor, energy, computer power, or other resources
that often exhibit variable tariffs over time. It leads to the
natural generalization of scheduling problems, in which using a time
slot incurs certain cost, varying over time, which we refer to as
{\em utilization cost} that must be paid in addition to the actual
scheduling cost. However natural and practicable this may seem,
there appears to be very little theoretical research on such
scheduling models. The only work we are aware of is by Wan and
Qi~\cite{WanQi2010},  Kulkarni and Munagala~\cite{KulkarniM2012},
Fang et al.~\cite{FangEtAl2016} and Chen and Zhang~\cite{ChenZ18}, where variable tariffs concern the
cost of labor or the cost of energy.

The goal of this paper is to expedite the theoretical understanding
of fundamental scheduling problems within the framework of
time-varying costs or tariffs. We contribute optimal polynomial-time
algorithms and best possible approximation algorithms for the
fundamental scheduling objectives of minimizing the sum of
weighted completion times and the makespan.

\subsection{Problem definition}

We first describe the underlying classical scheduling problems. We
are given a set of jobs $J:= \left\{ 1,\ldots,n\right\}$ where every
job $j\in J$ has a given processing time $p_j\in \mathbb{N}$
and possibly a weight $w_j \in \mathbb{Q}_{\geq 0}$.  The
objective is to find a preemptive schedule on a single machine
such that the total (weighted) completion time, $\sum_{j\in
J}w_jC_j$, is minimized; here $C_j$ denotes the completion time of
job $j$.
Preemption means that the processing of a job may be
interrupted at any time and can continue at any time later at no
additional cost. In the three-field scheduling notation
\cite{graham1979optimization}, this problem is denoted
as~\abc{1}{pmtn}{\sum w_jC_j}.
We also consider makespan minimization on unrelated machines,
\abc{R}{pmtn}{C_{\max}}, where we are given a set of machines~$M$,
and each job $j\in J$ has an individual processing time $p_{ij}\in
\mathbb{N}$ for running on machine $i\in M$. The objective is to
find a preemptive schedule that minimizes the makespan, that is, the
completion time of the latest job.

In this paper, we consider a generalization of these scheduling
problems within a time-of-use tariff model. We assume
that time is discretized into unit-size time slots. We are given a
tariff or cost function $e:\mathbb{N} \rightarrow \mathbb{Q}_{\geq
0}$, where~$e(t)$ denotes the \tff\ for processing job(s) at time
slot $[t,t+1)$. We assume that $e$ is piecewise constant. I.e., we
assume that the time horizon is partitioned into given intervals
$I_k =\left[s_k,d_k\right)$ with $s_k,d_k \in \mathbb{N}$,
$k=1,\ldots,K$, within which $e$ has the same value $e_k$. To ensure
feasibility, we assume that $d_K\geq \sum_{j\in J} \min_{i\in
M} p_{ij}$.

Given a schedule $\Sch$, let $y(t)$ be a binary variable indicating
if any processing is assigned to time slot~$\left[t,t+1\right)$. The
{\em utilization cost} of $\Sch$ is $E(\Sch)=\sum_{t} e(t) y(t)$.
That means, for any time unit that is used in $\Sch$ we pay the full
\tff, even if the unit is only partially used.  We also
emphasize that in the makespan problem in which we have multiple
machines, a time slot paid for can be used by all machines. This
models applications in which paying for a time unit on a resource
gives access to all units of the resource, e.g., all processors on a
server.

The overall objective is to find a schedule that minimizes the
scheduling objective, $\sum_{j\in J}w_jC_j$ resp. $C_{\max}$, {\em
plus} the utilization cost $E$. We refer to the resulting problems
as \abc{1}{pmtn}{\sum w_jC_j + E} and \abc{R}{pmtn}{C_{\max} + E}.
We emphasize that the results in this paper also hold if we minimize
any convex combination of the scheduling and utilization cost.

\subsection{Related work} 
Scheduling with time-of-use tariffs (aka variable time
slot cost) has been studied explicitly by Wan and
Qi~\cite{WanQi2010}, Kulkarni and Munagala~\cite{KulkarniM2012},
Fang et al.~\cite{FangEtAl2016} and Chen and Zhang~\cite{ChenZ18}. In their seminal paper, Wan and
Qi~\cite{WanQi2010} consider several {\em non-preemptive} single
machine problems, which are polynomial-time solvable in the
classical setting, such as minimizing the total completion time,
lateness, and total tardiness, or maximizing the weighted number of
on-time jobs. These problems are shown to be strongly NP-hard when
taking general tariffs into account, while efficient algorithms
exist for special monotone tariff functions.  In particular, the
problem \abc{1}{}{\sum C_j + E} is strongly NP-hard, and it is
efficiently solvable when the tariff function is increasing or
convex non-increasing~\cite{WanQi2010}. Practical applications,
however, often require non-monotone tariff functions, which lead to
wide open problems in the context of preemptive and non-preemptive
scheduling. In this paper, we answer complexity and approximability
questions for fundamental preemptive scheduling problems.

Kulkarni and Munagala~\cite{KulkarniM2012} focus on a relevant but
different problem in an {\em online setting}, namely online
flow-time minimization using resource augmentation. Their main
result is a scalable algorithm that obtains a constant performance
guarantee when the machine speed is increased by a constant factor
and there are only two distinct unit tariffs. They also show that,
in this online setting, for an arbitrary number of distinct unit
tariffs there is no constant speedup-factor that allows for a
constant approximate solution. For the problem considered in this
paper, offline scheduling without release dates, Kulkarni and
Munagala~\cite{KulkarniM2012} observed a relation to universal
sequencing on a machine of varying speed~\cite{Epstein_etal2012}
which implies the following results: a
pseudo-polynomial~$4$-approximation for \abc{1}{pmtn}{\sum w_jC_j +
E}, which gives an optimal solution in case that all weights are
equal, and a constant approximation in quasi-polynomial time for a
constant number of distinct tariffs or when using a machine that is
processing jobs faster by a constant factor.

Fang et al.~\cite{FangEtAl2016} study scheduling on a single machine under
time-of-use electricity tariffs. They do not take the scheduling cost
into account, but only the energy cost. It also differs from our approach
since the schedule is made by the provider and not by the user of the energy.
In their model the time horizon is divided into $K$ regions, each of which has
a cost $c_k$ per unit energy. For processing jobs the dynamic variable speed model is used; that
is, the energy consumption is $s^{\alpha}$ per unit time if jobs are
run at speed $s$, whence, within region $k$, the energy cost is
$s^{\alpha} c_k$.
The objective is to minimize energy cost such that all jobs are scheduled
within the $K$ regions.
They prove that the non-preemptive case is NP-hard and give a non-constant approximation, and for the preemptive case, they give a polynomial-time algorithm.

Chen and Zhang~\cite{ChenZ18} consider non-preemptive scheduling on a single machine so as to minimize the total \reservation\ cost under certain scheduling feasibility constraints such as a common deadline for all jobs or a bound on the maximum lateness, maximum tardiness, maximum flow-time, or sum of completion times. 
They define a \emph{valley} to be a cost interval $I_k$ that has smaller cost than its neighboring intervals and show the following.
General tariffs lead to a strongly NP-hard problem for any of the just mentioned constraints, and even very restricted tariff functions with more than one valley result in NP-hard problems that are not approximable within any constant factor.
The problem with a common deadline on the job completion times is shown to admit a pseudo-polynomial time algorithm when having two valleys, a polynomial time algorithm for tariff functions with at most one valley, and an FPTAS if there are at most two valleys and $\max_k e_k / \min_k e_k$ is bounded.  
For the other mentioned constraints, they also present polynomial time algorithms when having no more than one valley, where the problem with a bound on the sum of completion times requires the number of cost intervals, here $K$, to be fixed.

The general concept of taking into consideration additional
(time-dependent) cost for resource utilization when scheduling has
been implemented differently in other models. We mention the area of
energy-aware scheduling, where energy consumption is taken into
account (see~\cite{Albers10} for an overview). Further, the area of scheduling
with generalized non-decreasing (completion-) time dependent cost
functions, such as minimizing $\sum_j w_jf(C_j)$,
e.g.~\cite{Epstein_etal2012,MegowV2013,HoehnJ15}, or even more
general job-individual cost functions $\sum_j f_j(C_j)$,
e.g.~\cite{bansalP14,HoehnMW14,CheungS11,CheungEtAl2017} has
received quite some attention.  Our model differs fundamentally from
those models since our cost function may decrease with time. In
fact, delaying the processing in favor of cheaper time slots may
decrease the overall cost.  This is not the case in the
above-mentioned models. Thus, in our framework we have the
additional dimension in decision-making of selecting the time slots
that shall be utilized.

Finally, we point out  some similarity between our model and {\em
scheduling on a machine of varying speed}, which (with $\sum_j
w_jC_j$ as objective function) is an equivalent statement of the
problem of minimizing $\sum_j w_jf(C_j)$ on a single machine with
constant speed~\cite{Epstein_etal2012,MegowV2013,HoehnJ15}. We do
not see any mathematical reduction from one problem to the other.
However, it is noteworthy that  the independently studied problem of
scheduling with {\em non-availability periods}, see e.g. the survey
by Lee~\cite{Lee2004}, is a special case of both the varying-speed
and the time-varying tariff model. Indeed, machine non/availability
can be expressed either by $0/1$-speed or equivalently by $\infty/0$
tariff. Results shown in this context imply that our problem
\abc{1}{pmtn}{\sum w_j C_j + E} is strongly NP-hard, even if there
are only two distinct \tffs~\cite{WangSC2005}.

\subsection{Our contribution}

We present new optimal algorithms and best-possible approximation
results, unless P=NP, for the generalization of basic
scheduling problems to a framework with time-varying \tffs.

One of our results is a rather straightforward optimal
polynomial-time algorithm for the problem \abc{R}{pmtn}{C_{\max}+E}
(Section~\ref{sec:Cmax}): We design a procedure that selects the
optimal time slots to be utilized, given that we know their
optimal {\em number}. That number can be
determined by solving the scheduling problem {\em without} utilization
cost, which can be done in polynomial time by solving a
linear program~\cite{LawlerL1978}.

Whereas minimizing makespan plus utilization cost appears to be
efficiently solvable even in the most general machine model, the
objective of minimizing the total weighted completion time raises
significant complications. Our results on this objective concern
single-machine problems (Section~\ref{sec:GenericAlg}). We present
an algorithm that computes for a given ordered set of jobs an
optimal choice of time slots to be used. We derive this by first
showing structural properties of an optimal schedule, which we then
exploit together with a properly chosen potential function in a
dynamic program yielding polynomial running time. Based on this
algorithm, we show that the unweighted problem \abc{1}{pmtn}{\sum
C_{j} + E} can be solved in polynomial time and that it allows
almost directly for a fully polynomial $(4+\eps)$-approximation
algorithm for the weighted version \abc{1}{pmtn}{\sum w_j C_{j} +
E}, for which a pseudo-polynomial $4$-approximation was
observed by Kulkarni and Munagala~\cite{KulkarniM2012}. While
pseudo-polynomial time algorithms are relatively easy to
derive, it is quite remarkable that our DP's running time is
polynomial in the input, in particular, independent of $d_K$.

In Section~\ref{sec:PTAS}, we significantly improve the approximation result for the weighted problem by designing
a polynomial-time algorithm that computes for
any fixed $\eps$ a $(1+\eps)$-approximate schedule for
\abc{1}{pmtn}{\sum w_jC_j + E}, that is, we give a polynomial-time approximation
scheme (PTAS).  Unless P=NP, our algorithm is best possible, since the
problem is strongly NP-hard even if there are only two different
\tffs~\cite{WangSC2005}.

Our approach is inspired by a recent PTAS for scheduling on a
machine of varying speed~\cite{MegowV2013} and it uses some of its
properties. As discussed before, we do not see a formal mathematical
relation between these two seemingly related problems which allows
to apply the result from~\cite{MegowV2013} directly. The key is a
dual view on scheduling: instead of directly constructing a schedule
in the time-dimension, we first construct a dual scheduling solution
in the weight-dimension which has a one-to-one correspondence to a
true schedule.  We design an exponential-time dynamic programming
algorithm which can be trimmed to polynomial time using techniques
known for scheduling with varying speed~\cite{MegowV2013}.

For both the makespan and the min-sum problem, job preemption is
crucial for obtaining constant worst-case performance ratios.
For non-preemptive scheduling, a straightforward reduction from
\textsc{2-Partition} shows that no approximation is possible, unless
P$=$NP, even if there are only two different \tffs, $0$ and~$\infty$.

Finally, we remark that in general it is not clear that a schedule
can be encoded polynomially in the input. However, for our
completion-time based minimization objectives, it is easy to observe
that if an algorithm utilizes $p$ unit-size time slots in an
interval of equal cost, then it utilizes the first $p$ slots within
this interval, which simplifies the structure and the output of an
optimal solution in a crucial way.

We start below with presenting the more involved results for the
problems with scheduling objective minimizing total (weighted)
completion time. The efficient algorithm for the makespan objective
is then presented in Section~\ref{sec:Cmax}.

\section{An optimal algorithm for minimizing total completion time}
\label{sec:GenericAlg}

In this section, we show how to solve the unweighted problem
\abc{1}{pmtn}{\sum C_{j} + E} to optimality. Our main result is as follows.

\begin{theorem}\label{thm:Ptime-no-w}
  There is a polynomial-time algorithm for \abc{1}{pmtn}{\sum C_j + E}.
\end{theorem}

An algorithm for the scheduling problem with time-of-use tariffs
has to make essentially two types of decisions: (i) which time slots
to use and (ii) how to schedule the jobs in these slots. It is not
hard to see that these two decisions can be handled separately. In
fact, the following observation on the optimal sequencing of jobs
holds independently of the utilization decision and follows from a
standard interchange argument.

\begin{observation}\label{obs:SPT}
  In an optimal schedule $\mathcal{S}^*$ for the problem
  \abc{1}{pmtn}{\sum C_{j} + E}, jobs are processed according to the
  Shortest Processing Time First (SPT) rule.
\end{observation}

Thus, in the remainder of the section we can focus on determining
which time slots to use. We design an algorithm that computes, for
any given (not necessarily optimal) scheduling sequence~$\sigma$, an
optimal \reservation\ decision for $\sigma$. In fact, we show our
structural result even for the more general problem in which jobs
have arbitrary weights.

\begin{theorem} \label{thm:comp-slots} Given an instance of
  \abc{1}{pmtn}{\sum w_j C_{j} + E} and an arbitrary processing
  sequence of jobs $\sigma$, we can compute an optimal \reservation\
  decision for~$\sigma$ in polynomial time.
\end{theorem}

Combining the optimal choice of time slots
(Theorem~\ref{thm:comp-slots}) with the optimal processing order SPT
(Observation \ref{obs:SPT}) immediately implies
Theorem~\ref{thm:Ptime-no-w}.

The remainder of the section is devoted to proving
Theorem~\ref{thm:comp-slots}. Thus, we choose any (not
necessarily optimal) order of jobs,
$\sigma=(1,\ldots,n)$, in which the jobs must be processed. We want
to characterize an optimal schedule $\mathcal{S}^*$ for $\sigma$,
that is, the optimal choice of time slots for scheduling~$\sigma$.
We firstly identify structural properties of an optimal solution.
Essentially, we give a full characterization which we can compute
efficiently by dynamic programming.

More precisely, we establish a closed form that characterizes
the relationship between the \tff\ of an \reserved\ slot and job
weights in an optimal solution. This relationship allows to
decompose an optimal schedule into a series of sub-schedules. Our
algorithm will first compute all possible sub-schedules and then use
a dynamic programming approach to select and concatenate suitable
sub-schedules.

In principle, an optimal schedule may preempt jobs at fractional
time points.  However, since time slots can only be paid for
entirely, any reasonable schedule uses the \reserved\ slots
entirely as long as there are unprocessed jobs.
It can be shown by a standard interchange argument that this
is also true if we omit the requirement that time slots must be
\reserved\ entirely; for details, see~\cite{rischke-thesis}.
(We remark that for the makespan problem with multiple machines considered in
Section~\ref{sec:Cmax} this is not true.)

\begin{lemma}\label{lemma:integrality}
  Allowing to pay for utilizing partial time slots,
  there is an optimal schedule $\mathcal{S}^*$
  for \abc{1}{pmtn}{\sum w_j C_{j} + E}
  in which all \reserved\
  time slots are entirely \reserved\ and jobs are preempted only at
  integral points in time.
\end{lemma}

Next, we split the optimal schedule $\mathcal{S}^*$ for the given job sequence $\sigma=(1,\ldots,n)$ into smaller sub-schedules. To that end we introduce the
concept of a \emph{split point}.
\begin{definition}[Split Point]
  Consider an optimal schedule $\mathcal{S}^*$ and the set of
  potential split points
  $\mathcal{P}:=\bigcup_{k=1}^{K} \left\{ s_k , s_k+1 \right\} \cup
  \{d_K\}$.
  Let $S_j$ and $C_j$ denote the start time and completion time of job
  $j$, respectively. We call a time point $t\in \mathcal{P}$ a split
  point for $\mathcal{S}^*$ if all jobs that start before $t$ also
  finish their processing not later than $t$, i.e., if
  $\left\{j\in J : S_j < t \right\} = \left\{ j\in J : C_j \leq
    t\right\}$.
\end{definition}

Given an optimal schedule $\mathcal{S}^*$, let
$0=\tau_1< \tau_2 < \cdots < \tau_\ell = d_K$ be the {\em maximal}
sequence of split points of $\mathcal{S}^*$, i.e. the sequence
containing all split points of $\mathcal{S}^*$.  We denote the
interval between two consecutive split points $\tau_x$ and
$\tau_{x+1}$ as {\em region}
$R_x^{\mathcal{S}^*} :=\left[\tau_x,\tau_{x+1}\right)$, for
$x=1,\ldots,\ell-1$.

Consider now any region $R_x^{\mathcal{S}^*}$ for an optimal
schedule $\mathcal{S}^*$ with $x\in \left\{1,\ldots,\ell-1\right\}$
and let $J_x^{\mathcal{S}^*} := \left\{ j \in J : S_j \in
R_x^{\mathcal{S}^*} \right\}$, the jobs that start and finish
within $R_x^{\mathcal{S}^*}$. Note that $J_x^{\mathcal{S}^*}$ might
be empty.  Among all optimal schedules we shall consider an optimal
solution $\mathcal{S}^*$ that minimizes the value
$\sum_{t=0}^{d_K-1} t\cdot y(t)$, where $y(t)$ is a binary variable
that indicates if time slot $\left[t,t+1\right)$ is \reserved\ or
not.

We observe that any job $j$ completing at the beginning of a cost
interval $I_k$, i.e. $C_j=s_k\in R_x^{\mathcal{S}^*}$ or
$C_j=s_k+1\in R_x^{\mathcal{S}^*}$, would make $s_k$ resp.\ $s_k+1$
a split point. Thus, no such job can exist.
\begin{observation}\label{obs:ComplTime}
  There is no job $j\in J_x^{\mathcal{S}^*}$ with
  $C_j \in R_x^{\mathcal{S}^*} \cap \mathcal{P}$.
\end{observation}

We say that interval $I_k$ is \emph{partially \reserved\ } if at
least one time slot in $I_k$ is \reserved, but not all.

\begin{lemma}\label{lemma:AtMostOnePI}
  There exists an optimal schedule $\mathcal{S}^*$ in which for all
  $x=1,\ldots,\ell-1$
  at most
  one interval is partially \reserved\ in $R_x^{\mathcal{S}^*}$.
\end{lemma}
\begin{proof}
  By contradiction, suppose that there is more than one partially
  \reserved\ interval in $R_x^{\mathcal{S}^*}$.  Consider any two such
  intervals $I_h$ and $I_{h'}$ with $h<h'$, and all intermediate
  intervals \reserved\ entirely or not at all.  Let
  $\left[t_h,t_h+1\right)$ and $\left[t_{h'},t_{h'}+1\right)$ be the
  last \reserved\ time slot in $I_h$ and $I_{h'}$, respectively.  If we
  \reserve\ $\left[t_{h'}+1,t_{h'}+2\right)$ instead of
  $\left[t_h,t_h+1\right)$, then the difference in cost is
  $\delta_1 := e_{h'}-e_h + \sum_{j\in J'} w_j$ with
  $J':= \left\{j\in J: C_j \in \bigcup_{k=h+1}^{h'} I_k
  \right\}$ because all jobs in $J'$ are delayed by exactly one time
  unit. This is true since by Observation~\ref{obs:ComplTime}
  no job finishes at $d_k=s_{k+1}$ for  any $k$.
  If we \reserve\ $\left[t_{h}+1,t_{h}+2\right)$ instead of
  $\left[t_{h'},t_{h'}+1\right)$, then the difference in cost is
  $\delta_2 := e_{h}-e_{h'} - \sum_{j\in J'} w_j$, again using
  Observation~\ref{obs:ComplTime} to assert that no job finishes at
  $s_k+1$ for any $h+1\leq k \leq h'$.  Since $\delta_1=-\delta_2$ and
  $\mathcal{S}^*$ is an optimal schedule, it must hold that
  $\delta_1=\delta_2 = 0$.  This, however, implies that there is
  another optimal schedule with earlier used time slots which contradicts our assumption
  that $\mathcal{S}^*$ minimizes the value $\sum_{t=0}^{d_K-1} t \cdot
  y(t)$. \qed
\end{proof}

The next Lemma characterizes the time slots that are used within a region.
Let~$e_{\max}^j$ be the maximum \tff\ spent for job $j$ in~$\mathcal{S}^*$.  Furthermore,
let~$\Delta_x := \max_{j\in J_x^{\mathcal{S}^*}} (e_{\max}^j +
\sum_{j' < j} w_{j'})$
and let $j_x$ be the last job (according to sequence $\sigma$) that
achieves~$\Delta_x$.
Suppose, there are $b\geq0$ jobs before and
$a\geq 0$ jobs after job $j_x$ in~$J_x^{\mathcal{S}^*}$. The following
lemma gives for every job $j \in J_x^{\mathcal{S}^*}\setminus \{j_x\}$
an upper bound on the \tff\ spent in the interval
$\left[S_j,C_j\right)$.

\begin{lemma}\label{lemma:Bounds}
  Consider an optimal schedule $\mathcal{S}^*$ for a given job permutation $\sigma$. For any job
  $j\in J_x^{\mathcal{S}^*}\setminus \{j_x\}$ a slot
  $\left[t,t+1\right) \in \left[S_j,C_j\right)$ is \reserved\ if and
  only if the tariff $e(t)$ of $\left[t,t+1\right)$ satisfies the following upper
  bound:
  \begin{equation*}
    e(t) \ \leq\
    \begin{cases}
      e_{\max}^{j_x} + \sum_{j'=j}^{j_x -1} w_{j'}, & \forall j: j_x -b \leq j <  j_x\\[0.5ex]
     e_{\max}^{j_x} - \sum_{j'=j_x}^{j-1} w_{j'} & \forall j: j_x < j \leq j_x+a\,.
   \end{cases}
  \end{equation*}
\end{lemma}
\begin{proof}
  Consider any job $j:=j_x -\ell$ with $0< \ell \leq b$.  Suppose
  there is a job $j$ for which a slot is \reserved\ with cost (\tff)
  $e_{\max}^{j} > e_{\max}^{j_x} + \sum_{j'=j}^{j_x -1} w_{j'}$.  Then
  $e_{\max}^{j} + \sum_{j'<j} w_{j'} > e_{\max}^{j_x} + \sum_{j'<j_x}
  w_{j'}$,
  which is a contradiction to the definition of job $j_x$.  Thus,
  $e_{\max}^{j}\leq e_{\max}^{j_x} + \sum_{j'=j}^{j_x -1} w_{j'}$.

  Now suppose that there is a slot
  $\left[t,t+1\right) \in \left[S_j,C_j\right)$ with cost
  $e(t) \leq e_{\max}^{j_x} + \sum_{j'=j}^{j_x -1} w_{j'}$ that is not
  \reserved.  There must be a slot
  $\left[t',t'+1\right) \in \left[S_{j_x},C_{j_x}\right)$ with cost
  exactly $e_{\max}^{j_x}$.  If we \reserve\ slot $\left[t,t+1\right)$
  instead of $\left[t',t'+1\right)$, then the difference in cost is
  non-positive, because the completion times of at least $\ell$ jobs
  ($j=j_x-\ell, \ldots, j_x-1$ and maybe also $j_x$) decrease by one.
  This contradicts either the optimality of $\mathcal{S}^*$ or our
  assumption that $\mathcal{S}^*$ minimizes
  $\sum_{t=0}^{d_K-1} t \cdot y(t)$.

  The proof of the statement for any job $j_x +\ell$ with
  $0< \ell \leq a$ follows a similar argument, but now using the fact
  that for every job $j:=j_x+\ell$ we have
  $e_{\max}^{j} < e_{\max}^{j_x} - \sum_{j'=j_x}^{j-1} w_{j'}$,
  because $j_x$ was the last job with
  $e_{\max}^j + \sum_{j' < j} w_{j'} =\Delta_x$. \qed
\end{proof}

\begin{corollary}\label{cor:FullUtil}
    If the interval $[S_j,C_j)$ for processing
    a job $j \in J_x^{\mathcal{S}^*}\setminus \{j_x\}$ intersects
    interval $I_k$ but job~$j$ does not complete in $I_k$, i.e.,
    $C_j>d_k$, then all time slots in $I_k$ are fully utilized.
\end{corollary}

To decide on an optimal utilization decision for the
sub-schedule of the jobs in $R_x^{\mathcal{S}^*}$, we need the
following two lemmas.

\begin{lemma}\label{lemma:OnePartInt}
  If there is a partially \reserved\ interval $I_k$ in region
  $R_x^{\mathcal{S}^*}$, then (i) $I_k$ is the last interval of $R_x^{\mathcal{S}^*}$,
  or (ii) $j_x$ is the last job being processed in $I_k$ and $e_k=e_{\max}^{j_x}$.
\end{lemma}
\begin{proof}
Suppose there exists a partially \reserved\ interval $I_k$ in region
$R_x^{\mathcal{S}^*}$. Suppose $j$ with $j \neq {j_x}$ is the
last job that is processed in $I_k$, hence (ii) does not hold.
Then either $C_j < d_k$, in which case $d_k=s_{k+1}$
is a split point and thus $I_k$ is the last interval in the region,
whence (i) is true. Or, we are in the situation of
Corollary~\ref{cor:FullUtil} and have a contradiction,
because then $I_k$ must be fully utilized.

Now suppose $j_x$ is the last job being processed in $I_k$. If
$C_{j_x} < d_k$, then again $I_k$ is the last interval in the
region. Otherwise $C_{j_x} \notin I_k$. If $e_k =
e_{\max}^{j_x}$, then case $(ii)$ of the lemma holds. If not,
by definition of
$e_{\max}^{j_x}$ we have $e_k < e_{\max}^{j_x}$. By
optimality of $\mathcal{S}^*$, interval $I_k$ comes after the last
utilized ``expensive'' interval with cost $e_{\max}^{j_x}$. Hence,
job $j_x$ is processed in an expensive interval, then in $I_k$
and is completed in yet another interval. But then we can utilize an
extra time slot in $I_k$ instead of a time slot in the expensive
interval, without increasing the completion time. This contradicts
optimality, and, hence, $e_k = e_{\max}^{j_x}$, which completes the
proof. \qed
\end{proof}

\begin{lemma}\label{lemma:NoEmpty}
  There exists an optimal schedule $\mathcal{S}^*$ for a given job
  permutation $\sigma$ with the following property. If the last interval
  $I_k$ of a region $R_x^{\mathcal{S}^*}\!$ is only partially \reserved\,
  then all time slots in $\left[S_{j_x},C_{j_x}\right)$ with cost at most
  $e_{\max}^{j_x}$ are utilized.
\end{lemma}
\begin{proof}
Recall that $j_x + a$ is the last job being processed in the region, and hence, it is the last job processed in the partially \reserved\ interval $I_k$.

  Suppose there is a time slot
  $\left[t,t+1\right) \in \left[S_{j_x},C_{j_x}\right)$ with cost at
  most $e_{\max}^{j_x}$ that is not \reserved.  If we \reserve\   $\left[t,t+1\right)$ instead of the last \reserved\ slot in $I_k$,
  then the difference in cost is
  $\delta_1:=e(t) -e_k - \sum_{j=j_x}^{j_x+a} w_{j}$.
  On the other hand, if we
  \reserve\ one additional time slot in $I_k$ instead of a time slot
  in $\left[S_{j_x},C_{j_x}\right)$ with cost $e_{\max}^{j_x}$, then the difference in cost is
  $\delta_2 := e_k - e_{\max}^{j_x} + \sum_{j=j_x}^{j_x+a} w_{j}$.
  We consider an optimal schedule $\mathcal{S}^*$, thus
  $\delta_1 \geq 0$ and $\delta_2 \geq 0$ which implies that
  $\delta_1 + \delta_2 = e(t) - e_{\max}^{j_x} \geq 0$.  This is a
  contradiction if $e(t) < e_{\max}^{j_x}$.  If
  $e(t) = e_{\max}^{j_x}$, then $\delta_1 = -\delta_2 = 0$, because we
  consider an optimal schedule $\mathcal{S}^*$.  This, however,
  contradicts our assumption that $\mathcal{S}^*$ minimizes the value
  $\sum_{t=0}^{d_K-1} t \cdot y(t)$. \qed
\end{proof}

We now show how to construct an optimal partial schedule for a given
ordered job set in a given region in polynomial time.

\begin{lemma}\label{lem:partial-schedule}
  Given a region $R_x$ and an ordered job set $J_x$, we can find
  in polynomial time an optimal \reservation\ decision for
  scheduling
  $J_x$ within the region
  $R_x$, which does not contain any other split points than $\tau_x$
  and $\tau_{x+1}$, the boundaries of $R_x$.
\end{lemma}
\begin{proof}
  Given $R_x$ and $J_x$, we guess the optimal
  combination~$\left(j_x,e_{\max}^{j_x}\right)$, i.e., we enumerate
  over all $nK$ combinations and choose eventually the best solution.

  We firstly assume that a partially \reserved\ interval exists and it
  is the last one in $R_x$ (case (i) in
  Lemma~\ref{lemma:OnePartInt}).  Based on the characterization in
  Lemma~\ref{lemma:Bounds} we find in polynomial time the slots to be
  \reserved\ for the jobs $j_x-b,\ldots,j_x-1$.  This defines
  $C_{j_x-b},\ldots,C_{j_x-1}$. Then starting job $j_x$ at time
  $C_{j_x-1}$, we check intervals in the order given and \reserve\ as
  much as needed of each next interval $I_h$ if and only if
  $e_h\leq e_{\max}^{j_x}$, until a total of $p_{j_x}$ time slots have
  been \reserved\ for processing~$j_x$. Lemma~\ref{lemma:NoEmpty}
  justifies to do that.  This yields a completion time $C_{j_x}$.
  Starting at $C_{j_x}$, we use again Lemma~\ref{lemma:Bounds} to find
  in polynomial time the slots to be \reserved\ for processing the jobs
  $j_x+1,\ldots,j_x+a$.  This gives $C_{j_x+1},\ldots,C_{j_x+a}$.

  Now we assume that there is no partially \reserved\ interval or we are
  in case~(ii) of Lemma~\ref{lemma:OnePartInt}.  Similar to the case
  above, we find in polynomial time the slots that $\mathcal{S}^*$
  \reserves\ for the jobs $j_x-b,\ldots,j_x-1$ based on
  Lemma~\ref{lemma:Bounds}.  This defines
  $C_{j_x-b},\ldots,C_{j_x-1}$.  To find the slots to be \reserved\ for
  the jobs $j_x+1,\ldots,j_x+a$, in this case, we start at the end of
  $R_x$ and go backwards in time. We can start at the end of $R_x$
  because in this case the last interval of $R_x$ is fully \reserved.
  This gives $C_{j_x+1},\ldots,C_{j_x+a}$.  Job $j_x$ is thus to be
  scheduled in $\left[C_{j_x-1}, S_{j_x+1}\right)$.  In order to find
  the right slots for $j_x$ we solve a makespan problem in the
  interval $\left[C_{j_x-1}, S_{j_x+1}\right)$, which can be done in
  polynomial time (Theorem~\ref{thm:SingleMachMakespan}) and gives a
  solution that cannot be worse than what an optimal schedule
  $\mathcal{S}^*$ does.

  If anywhere in both cases the \reserved\ intervals can not be made
  sufficient for processing the job(s) for which they are intended, or
  if scheduling the jobs in the \reserved\ intervals creates any
  intermediate split point, then this
  $\left(j_x,e_{\max}^{j_x}\right)$-combination is rejected.  Hence,
  we have computed the optimal schedules over all $nK$ combinations of
  $\left(j_x,e_{\max}^{j_x}\right)$ and over both cases of
  Lemma~\ref{lemma:OnePartInt} concerning the position of the
  partially \reserved\ interval. We choose the schedule with minimum
  total cost and return it with its value.  This completes the proof.
  \qed
\end{proof}

Now we are ready to prove our main theorem.
\begin{proof}[Proof of Theorem~\ref{thm:comp-slots}]
  We give a dynamic program. Assume jobs are indexed according to the order given by $\sigma$.  We define a state $(j,t)$, where $t$ is a potential split point $t\in \mathcal{P}$ and $j$ is a job from the job set $J$, and a dummy job $0$. The value of a state, $Z(j,t)$,  is the optimal scheduling cost plus \reservation\ cost for completing jobs $1,\ldots, j$ by time $t$. We apply the following recursion:
  \begin{align*}
    Z(j, t) &= \min\!\bigg\{ Z(j'\!,t') + z\big(\big\{j'\!+1,\ldots,j\big\},
    [t',t)\big) \, | \,  t',t\in \mathcal{P}, t'<  t,  j',j \in J,  j'\leq j \bigg\},\\
    Z(0, t) &= 0,\quad \textup{ for any } t,\\
    Z(j, s_1) &= \infty ,\quad \textup{ for any } j>0,
  \end{align*}
  where $z\big(\big\{j'+1,\ldots,j\big\},[t',t)\big)$ denotes the
  value of an optimal partial schedule for job set
  $\{j'+1,j'+2,\ldots,j\big\}$ in the region $[t',t)$, or $\infty$ if
  no such schedule~exists. In case $j=j'$ there is no job to be scheduled in the interval $[t',t)$, whence we set $z\big(\big\{j'+1,\ldots,j\big\},[t',t)\big)=0$. This models the option of leaving regions empty.

  An optimal partial schedule can be computed in polynomial time as we have shown in Lemma~\ref{lem:partial-schedule}.   Hence, we compute $Z(j, t)$ for all $O(nK)$ states in polynomial
  time, which concludes the~proof.  \qed
\end{proof}

\subsubsection*{Remark: A simple $(4+\epsilon)$-approximation for the weighted problem.}

{\noindent It is worth mentioning that the characterization of
an optimal \reservation\ decision
above~(Theorem~\ref{thm:comp-slots}) can
be used to obtain a simple $(4+\eps)$-approximation for the {\em
weighted} problem \abc{1}{pmtn}{\sum w_j C_{j} + E}.

For the weighted problem, there may not exist a job sequence that is
universally optimal for {\em all} \reservation\
decisions~\cite{Epstein_etal2012}.  However, in the context of
scheduling on an unreliable machine there has been shown a
polynomial-time algorithm that computes a
universal~$(4+\eps)$-approximation~\cite{Epstein_etal2012}. More
precisely, the algorithm constructs a sequence of jobs which
approximates the scheduling cost for any \reservation\ decision
within a factor at most $4+\eps$.

Consider an instance of problem \abc{1}{pmtn}{\sum w_j C_{j} + E} and
compute such a universally $(4+\eps)$-approximate sequence $\sigma$. Applying
Theorem~\ref{thm:comp-slots} to $\sigma$, we obtain a schedule
$\mathcal{S}$ with an optimal \reservation\ decision for $\sigma$. Let
$\mathcal{S'}$ denote the schedule which we obtain by changing the
\reservation\ decision of $\mathcal{S}$ to the \reservation\ in an optimal
schedule $\mathcal{S}^*$ (but keeping the scheduling sequence
$\sigma$). The schedule $\mathcal{S'}$ has cost no less than the
original cost of $\mathcal{S}$. Furthermore, given the \reservation\ decision in the optimal solution $\mathcal{S}^*$, the sequence
$\sigma$ approximates the scheduling cost of $\mathcal{S}^*$ within a
factor of $4+\eps$.  This gives the following result.

\medskip
\begin{corollary}
  There is a $(4+\eps)$-approximation algorithm for \abc{1}{pmtn}{\sum w_j
    C_j + E}.
\end{corollary}

This result is superseded by the PTAS presented in the next section.

\section{A PTAS for minimizing the total weighted completion time}
\label{sec:PTAS}

The main result of this section is a polynomial time approximation scheme for
minimizing the total weighted completion time with time-varying
\reservation\ cost.

\begin{theorem}\label{thm:PTAS}
  There is a polynomial-time approximation scheme for
  \abc{1}{pmtn}{\sum w_j C_j + E}.
\end{theorem}

In the remainder of this section we describe some preliminaries,
present a dynamic programming~(DP) algorithm with exponential running
time, and then we argue that the running time can be reduced to
polynomial time.  As noted in the introduction, our approach is
inspired by a PTAS for scheduling on a machine of varying
speed~\cite{MegowV2013}, but a direct application does not seem
possible.

\subsection{Preliminaries and scheduling in the weight-dimension}
\label{sec:prelim-ptas}

We describe a schedule $\mathcal{S}$ not in terms of completion times
$C_j\!\left( \mathcal{S}\right)$, but in terms of the remaining weight
function $W^{\mathcal{S}}\!\left( t\right)$ which, for a given
schedule $\Sch$, is defined as the total weight of all jobs not
completed by time~$t$. Notice that, by definition, $W^{\mathcal{S}}(t)$ is right-continuous.
Based on the remaining weight function we can
express the cost for any schedule $\mathcal{S}$ as
$$\int_0^\infty W^{\mathcal{S}}\!\left( t\right) \mathrm{d}t = \sum_{j\in J} w_j
C_j\!\left( \mathcal{S}\right).$$
This has a natural interpretation in the standard \mbox{2D-Gantt}
chart, which was originally introduced in~\cite{EastmanEI1964}.

For a given \reservation\ decision, we follow the idea of\cite{MegowV2013} and implicitly describe the completion time of a job
$j$ by the value of the function $W^{\mathcal{S}}$ at the time that $j$ completes.
This value is referred to as the \emph{starting weight} $S_j^w$ of job
$j$.  In analogy to the time-dimension, the value
$C_j^w := S_j^w +w_j$ is called \emph{completion weight} of job
$j$. When we specify a schedule in terms of the remaining weight
function, then we call it a \emph{weight-schedule}, otherwise a
\emph{time-schedule}.  Other terminologies, such as feasibility and
idle time, also translate from the time-dimension to the
weight-dimension.  A weight-schedule is called \emph{feasible} if no
two jobs overlap and the machine is called \emph{idle in
  weight-dimension} if there exists a point $w$ in the
weight-dimension with $w\notin \left[S_j^w, C_j^w\right]$ for all jobs
$j\in J$.

A weight-schedule together with a \reservation\ decision can be
translated into a time-schedule by ordering the job in decreasing
order of completion weights and scheduling them in this order in the
time-dimension in the \reserved\ time slots.  For a given \reservation\ decision, consider a weight-schedule $\mathcal{S}$ with completion
weights $C_1^w > \cdots > C_n^w > C_{n+1}^w:=0$ and the corresponding
completion times $0=:C_0 < C_1 < \cdots < C_n$ for the jobs
$j=1,\ldots,n$.  We define the \emph{(scheduling) cost of a weight-schedule} $\mathcal{S}$ as
$\sum_{j=1}^n \left( C_{j}^w -C_{j+1}^w \right) C_j$.  This value equals
$\sum_{j=1}^n \pi_j^{\mathcal{S}} C_j^w$, where
$\pi_j^{\mathcal{S}}:= C_j -S_j$, if and only if there is no idle
weight.  If there is idle weight, then the cost of a weight-schedule
can only be greater, and we can safely remove idle weight without
increasing the scheduling cost~\cite{MegowV2013}.
Figure~\ref{fig:weightSchedule} illustrates this fact.

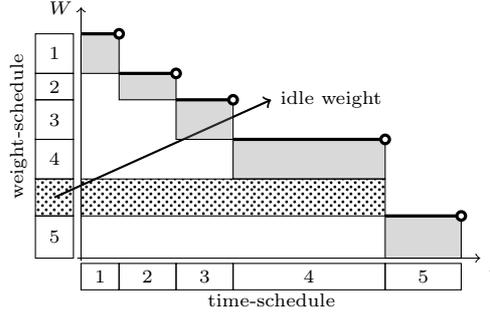
\begin{figure}
\begin{center}
\begin{tikzpicture}[scale=0.5,yscale=.7,font=\scriptsize]
\tikzstyle{disc}=[draw,fill=white,circle,very thick,inner sep=0.4mm]
\tikzstyle{cont}=[fill=black,circle,very thick,inner sep=0.4mm]

\def\N{9}

\def\x{{0,1,2.5,4,8,10}}%
\def\y{{8.5,7,6,4.5,3,1.6,0}}%

 \draw[->] (-.1,0) -- (10.5,0) node[below right] {$t$};
 \draw[->] (0,-.1) -- (0,9.5) node[left] {$W$};

\begin{scope}[yshift=-1.2cm]
 \foreach \i/\j in {0/1 , 1/2, 2/3 , 3/4 , 4/5} {%
  \pgfmathsetmacro{\xa}{\x[\i]}
  \pgfmathsetmacro{\xb}{\x[\i+1]}
  \draw (\xa,0) rectangle node {\j} (\xb,1);
 }
 \node at (5,-.4) {time-schedule};
\end{scope}

\begin{scope}[xshift=-1.2cm]
 \foreach \i in {0,...,3} {%
  \pgfmathsetmacro{\ya}{\y[\i]}
  \pgfmathsetmacro{\yb}{\y[\i+1]}
  \draw (0,\ya) rectangle node {\pgfmathparse{\i+1}\pgfmathprintnumber{\pgfmathresult}} (1,\yb);
 }
  \draw[pattern=crosshatch dots] (0,\y[4]) rectangle (1,\y[5]);
  \draw (0,\y[5]) rectangle node {5} (1,\y[6]);
  \node[rotate=90,above] at (0,5) {weight-schedule};
\end{scope}
\draw[->,thick] (-.7,2.3) -- (5,6) node[right] {idle weight};

 \foreach \i in {0,...,3} {%
  \pgfmathsetmacro{\xa}{\x[\i]}
  \pgfmathsetmacro{\xb}{\x[\i+1]}
  \pgfmathsetmacro{\ya}{\y[\i]}
  \pgfmathsetmacro{\yb}{\y[\i+1]}
  \draw[fill=gray,fill opacity=0.3] (\xa,\yb) rectangle (\xb,\ya);
 }
 \draw[fill=gray,fill opacity=0.3] (\x[4],\y[6]) rectangle (\x[5],\y[5]);
 \draw[pattern=crosshatch dots] (0,\y[4]) rectangle (\x[4],\y[5]);

 \foreach \i in {0,...,3} {%
  \pgfmathsetmacro{\xa}{\x[\i]}
  \pgfmathsetmacro{\xb}{\x[\i+1]}
  \pgfmathsetmacro{\ya}{\y[\i]}
  \draw[very thick]  (\xa,\ya)  -- (\xb,\ya) node[disc] {};
  }
 \draw[very thick]  (\x[4],\y[5])  -- (\x[5],\y[5]) node[disc] {};

\end{tikzpicture}
\end{center}
\caption{\textbf{2D-Gantt chart.} The x-axis shows a time-schedule, while the y-axis corresponds to $W(t) =  \sum_{C_j>t} w_j$ plus the idle weight in the corresponding weight-schedule \cite{MegowV2013}. \label{fig:weightSchedule}}
\end{figure}

Summarizing, a time-schedule implies a correspondent weight-schedule of the same cost. On the other hand, a weight-schedule plus a \reservation\ decision implies a time-schedule with a possibly smaller cost.

\subsection{Dynamic programming algorithm}
\label{sec:DP}
Let $\eps>0$.  Firstly, we scale the input parameters so that all job
weights $w_j$, $j=1,\ldots,n$, and all \tffs\ $e_k$,
$k=1,\ldots,K$, are non-negative integers.  Then, we apply standard
geometric rounding to the weights to gain more structure on the input,
i.e, we round the weights of all jobs up to the next integer power of
$(1+\varepsilon)$, by losing at most a factor $(1+\eps)$ in the
objective value.  Furthermore, we discretize the weight-space into
intervals of exponentially increasing size: we define intervals
$W\!I_u := [\left(1+\varepsilon\right)^{u-1},
\left(1+\varepsilon\right)^{u} )$
for $u = 1, \ldots, \nu$ with
$\nu := \lceil \log_{1+\varepsilon} \sum_{j\in J} w_j\rceil$.

Consider a subset of jobs $J' \subseteq J$ and a partial
weight-schedule of $J'$.  In the dynamic program, the set $J'$
represents the set of jobs at the beginning of a corresponding
weight-schedule, i.e., if $j\in J'$ and $k\in J\setminus J'$, then
\mbox{$C_j^w < C_k^w$}.  However, the jobs in $J'$ are scheduled at
the end in a corresponding time-schedule.  As discussed in
Section~\ref{sec:prelim-ptas}, a partial weight-schedule $\mathcal{S}$
for the jobs in $J \setminus J'$ together with a \reservation\ decision
for these jobs can be translated into a time-schedule.

Let
$\mathcal{F}_u := \{ J_u \subseteq J : \sum_{j\in J_u}w_j \leq \left(
  1+\varepsilon\right)^u \}$
for $u=1,\ldots,\nu$.  The set~$\mathcal{F}_u$ contains all the
possible job sets $J_u$ that can be scheduled in $W\!I_u$ or before.
Additionally, we define $\mathcal{F}_0$ to be the set that contains
only the set of all zero-weight jobs $J_0 := \{ j \in J : w_j
=0\}$.
The following observation allows us to restrict to simplified completion weights.

\begin{observation}\label{obs:CompletionWeight}
  Consider an optimal weight-schedule in which the set of jobs with
  completion weight in $W\!I_{u}$, $u\in \{1,\ldots,\nu \}$, is
  exactly $J_{u}\setminus J_{u-1}$ for some $J_{u}\in \mathcal{F}_{u}$
  and $J_{u-1} \in \mathcal{F}_{u-1}$.  By losing at most a factor
  $(1+\eps)$ in the objective value, we can assume that for all
  $u\in \{1,\ldots,\nu \}$ the completion weight of the jobs in
  $J_{u}\setminus J_{u-1}$ is exactly~$(1+\eps)^{u}$.
\end{observation}
The following observation follows from a simple interchange argument.
\begin{observation}\label{obs:ZeroWeightJobs}
   There is an optimal time-schedule in which $J_0$ is scheduled completely after all jobs
   in $J\setminus J_0$.
 \end{observation}

The dynamic program recursively constructs states $Z=[J_u,b,avg]$ and computes for every state a time point~$t(Z)$ with the following meaning.
A state $Z=[J_u,b,avg]$ with time point $t(Z)$ expresses that there is
a feasible partial time-schedule~$\mathcal{S}$ for the jobs in
$J \setminus J_u$ with $J_u \in \mathcal{F}_u$ together with a \reservation\
decision for the time interval~$[0,t(Z))$ with total \reservation\ cost at most $b$ and for which the
\emph{average scheduling cost}, i.e.,
\[\frac{1}{t(Z)} \cdot \int_0^{t(Z)} W^{\mathcal{S}}\!\left( t\right)
\mathrm{d}t,\]
is at most $avg$. We remark that even if $\mathcal{S}$ only schedules jobs in $J\setminus J_u$, the remaining weight function $W^{\mathcal{S}}$ still considers jobs in $J\setminus J_u$, and thus $W^{\mathcal{S}}(t(Z))= \sum_{j\in J\setminus J_u}w_j$. Also, $\mathcal{S}$ implies a weight-schedule for jobs in $J\setminus J_u$ where the completion weights belong to $[\sum_{j\in J_u} w_j,\sum_{j\in J} w_j]$.
Note that $avg \cdot t(Z)$ is an upper bound on the
total scheduling cost of~$\mathcal{S}$ and that the average scheduling
cost is non-increasing in time, because the remaining weight function
$W^{\mathcal{S}}\!\left( t\right)$ is non-increasing in time.
In the \emph{iteration for $u$}, we only consider states $[J_u,b,avg]$ with $J_u\in \mathcal{F}_u$.
The states in the iteration for~$u$ are created based on the states from the iteration for~$u+1$.
Initially, we only have the state $Z_\nu = [J,0,0]$
with $t(Z_\nu):=0$, we start the dynamic program with $u=\nu-1$,
iteratively reduce $u$ by one, and stop the process after the
iteration for $u=0$.  In the iteration for~$u$, the states together
with their time points are constructed in the following way.  Consider
candidate sets \mbox{$J_{u+1} \in \mathcal{F}_{u+1}$} and
\mbox{$J_{u} \in \mathcal{F}_{u}$} with $\mathcal{F}_u\subseteq \mathcal{F}_{u+1}$, a partial
time-schedule~$\mathcal{S}$ of $J \setminus J_{u}$, in which the set
of jobs with completion weight (in the correspondent weight-schedule) in $W\!I_{u+1}$ is exactly
$J_{u+1}\setminus J_{u}$ and the set of jobs later than $W\!I_{u+1}$ is exactly $J\setminus J_{u+1}$, two budgets $b_1,b_2$ with $b_1 \leq b_2$,
and two bounds on the average scheduling cost $avg_1,
avg_2$.
Let \mbox{$Z_1= [J_{u+1},b_1,avg_1]$} and $Z_2= [J_{u},b_2,avg_2]$ be
the corresponding states.  We know that there is a feasible partial
schedule for the job set $J \setminus J_{u+1}$ up to time $t(Z_1)$
having average scheduling cost at most $avg_1$ and \reservation\ cost at
most $b_1$.  By augmenting this schedule, we want to compute a minimum
time point $t(Z_1,Z_2)$ that we associate with the link between $Z_1$
and $Z_2$ so that there is a feasible partial schedule for
$J \setminus J_{u}$ that processes the jobs from
$J_{u+1}\setminus J_{u}$ in the interval $[t(Z_1),t(Z_1,Z_2))$, has
average scheduling cost at most $avg_2$, and \reservation\ cost at most
$b_2$.  That is, $t(Z_1,Z_2)$ is the minimum makespan if we start with
$Z_1$ and want to arrive at~$Z_2$.  For the computation of
$t(Z_1,Z_2)$, we use the following subroutine.

Using Observation~\ref{obs:CompletionWeight}, we approximate the area
under the remaining weight function $W^{\mathcal{S}}\!\left( t\right)$
for the jobs in $J_{u+1}\setminus J_{u}$ by
$(1+\eps)^{u+1} \cdot (t(Z_1,Z_2)-t(Z_1))$, where $t(Z_1,Z_2)$ is the
time point that we want to compute.  Approximating this area gives us
the flexibility to schedule the jobs in $J_{u+1}\setminus J_{u}$ in
any order.  However, we need that $avg_2 \cdot t(Z_1,Z_2)$ is an upper
bound on the integral of the remaining weight function by time
$t(Z_1,Z_2)$.  That is, we want that
\[ avg_2 \cdot t(Z_1,Z_2) \geq (1+\eps)^{u+1} \cdot t(Z_1,Z_2) +
t(Z_1)\cdot (avg_1- (1+\eps)^{u+1}).\]
Both the left-hand side and the right-hand side of this inequality are
linear functions in $t(Z_1,Z_2)$.  So, we can compute a smallest time
point $t^{LB}$ such that the right-hand side is greater or equal to
the left-hand side for all $t(Z_1,Z_2) \geq t^{LB}$.  If there is no
such $t^{LB}$, then we set $t(Z_1,Z_2)$ to infinity and stop the
subroutine.  Otherwise, we know that our average scheduling cost at
$t^{LB}$ or later is at most $avg_2$.  Let $E(p,[t_1,t_2))$ denote the
total cost of the $p$ cheapest slots in the time-interval~$[t_1,t_2)$.  We
compute the smallest time point $t(Z_1,Z_2) \geq t^{LB}$ so that the
set of jobs $J_{u+1}\setminus J_{u}$ can be feasibly scheduled in
$[t(Z_1),t(Z_1,Z_2))$ having \reservation\ cost not more than
$b_2 - b_1$.  That is, we set
\[t(Z_1,Z_2)=\min\left\{t \geq \max\{t(Z_1),t^{LB}\} :
  E(p(J_{u+1}\setminus J_{u}),[t(Z_1),t)) \leq b_2-b_1\right\}.\]
The time point~$t(Z_1,Z_2)$ can be computed in polynomial time by
applying binary search to the interval~$[\max\{t(Z_1),t^{LB}\},d_K)$,
since $E(p,[t_1,t_2))$ is a monotone function in $t_2$.

Given all possible states $[J_{u+1},b_1,avg_1]$ from the iteration for $u+1$,
the dynamic program enumerates for all these states all possible links
to states $[J_u,b_2,avg_2]$ from the iteration for $u$ fulfilling the above
requirement on the candidate sets $J_{u+1}$ and $J_u$, on the budgets
$b_1$ and $b_2$, and on the average scheduling costs $avg_1$ and
$avg_2$.  For any such possible link $(Z_1,Z_2)$ between states from
the iteration for $u+1$ and $u$, we apply the above subroutine and associate
the time point $t(Z_1,Z_2)$ with this link.  Thus, the dynamic program
associates several possible time points with a state
$Z_2 = [J_{u},b_2,avg_2]$ from the iteration for $u$.  However, we only keep
the link with the smallest associated time point $t(Z_1,Z_2)$ (ties
are broken arbitrarily) and this defines the time point $t(Z_2)$ that
we associate with the state $Z_2$.  That is, for a state $Z_2$ from
the iteration for $u$ we define
$t(Z_2) := \min\{t(Z_1,Z_2)\,|\,Z_1 \mbox{ is a state from the iteration for }
u+1\}$.

Let $E_{\max}$ be an upper bound on the total \reservation\ cost in an optimal solution, e.g., the total cost of the first
  $p(J)$ finite-cost time slots.
The dynamic
program does not enumerate all possible budgets but only a polynomial
number of them, namely budgets with integer powers of $(1+\eta_1)$
with $\eta_1 > 0$ determined later.  That is, for the budget on the
\reservation\ cost, the dynamic program enumerates all values in
\[
B:=\{0,1,(1+\eta_1),(1+\eta_1)^2,\ldots,(1+\eta_1)^{\omega_1}\}
\textup{ with } \omega_1 = \lceil \log_{1+\eta_1} E_{\max}\rceil.
\]

The value $\eta_1$ will be chosen so that $(1+\eta_1)^{\omega_1}\le (1+\varepsilon)$
and $\omega_1$ is polynomial (see proof of Lemma~\ref{lemma:DPApprox} for the exact definition).
Similarly, we observe that $(1+\eps)^{\nu}$ is an upper bound on the
average scheduling cost.  The dynamic program does also only enumerate
a polynomial number of possible average scheduling costs, namely
integer powers of $(1+\eta_2)$ with $\eta_2 > 0$ also determined later.  This means, for the
average scheduling cost, the dynamic program enumerates all values in
\[ AVG:=\{0,1,(1+\eta_2),(1+\eta_2)^2,\ldots,(1+\eta_2)^{\omega_2}\}
\textup{ with } \omega_2 = \lceil \nu \log_{1+\eta_2} (1+\eps)\rceil.
\]

As before, the value $\eta_2$ will be chosen so that
$(1+\eta_2)^{\omega_2}\le (1+\varepsilon)$ and $\omega_2$ is polynomial.
The dynamic program stops after the iteration for $u=0$.  Now, only
the set of zero-weight jobs is not scheduled yet.  For any state
$Z=[J_0,b,avg]$ constructed in the iteration for $u=0$, we append the
zero-weight jobs starting at time $t(Z)$ and \reserving\ the cheapest
slots, which is justified by Observation~\ref{obs:ZeroWeightJobs}.
We add the additional \reservation\ cost to $b$.  After this, we return the
state $Z=[J_0,b,avg]$ and its corresponding schedule, which can be
computed by backtracking and following the established links, with
minimum total cost $b+ avg\cdot t(Z)$.  With this, we obtain the
following result.
\begin{lemma}\label{lemma:DPApprox}
  The dynamic program computes a $(1+O(\eps))$-approximate
  solution.
\end{lemma}
\begin{proof}
  Consider an arbitrary iteration $u$ of the dynamic program and let $i=\nu-u$.
  We consider states $Z=[J_u,b,avg]$ with
  $J_u \in \mathcal{F}_u$, $b\in B$, and $avg \in AVG$ for which we
  construct the time points $t(Z)$.  Let
  $Z_1^* = [J_{u+1}^*,b_1^*,avg_1^*]$ and
  $Z_2^* = [J_u^*,b_2^*,avg_2^*]$ with
  $J_{u+1}^* \in \mathcal{F}_{u+1}$ and $J_u^* \in \mathcal{F}_{u}$ be
  the states that represent an optimal solution $\mathcal{S}^*$ for which the
  set of jobs with completion weight in $W\!I_{u+1}$ is exactly
  $J_{u+1}^*\setminus J_{u}^*$.  By
  Observation~\ref{obs:CompletionWeight}, we assume that also in $\mathcal{S}^*$
  the area under the remaining weight function
  $W^{\mathcal{S^*}}\!\left( t\right)$ for the jobs in
  $J_{u+1}^*\setminus J_{u}^*$ is approximated by
  $(1+\eps)^{u+1} \cdot (t(Z_2^*)-t(Z_1^*))$.  We now show the
  following.  The dynamic program constructs in iteration $i$ a state
  $Z=[J_u,b,avg]$ with $J_u \in \mathcal{F}_u$, $b\in B$, and
  $avg \in AVG$ such that
  \begin{enumerate}[(i)]
  \item $J_u = J_u^*$,
  \item $b \leq (1+\eta_1)^i \cdot b_2^*$,
  \item $avg \leq (1+\eta_2)^i \cdot avg_2^*$, and
  \item $t(Z) \leq t(Z_2^*)$.
  \end{enumerate}
  We prove this statement by induction on $i=1,\ldots,\nu$.  Consider the first
  iteration of the dynamic program, in which we consider states with
  job sets from $\mathcal{F}_{\nu-1}$.  Let
  $Z^*=[J_{\nu -1}^*,b^*,avg^*]$ be the state that corresponds to the
  optimal solution $\mathcal{S}^*$.  The dynamic program also considers the job
  set $J_{\nu -1}^*$.  Suppose, we \reserve\ the same slots that $\mathcal{S}^*$
  \reserves\ for the jobs in $J \setminus J_{\nu -1}^*$ in the interval
  $[0,t(Z^*))$.  Let~$b$ be the resulting \reservation\ cost after
  rounding $b^*$ up to the next value in $B$.  With this, we know that
  $b \leq (1+\eta_1) \cdot b^*$.  Furthermore, by our assumption, we
  know that the average scheduling cost of $\mathcal{S}^*$ up to time $t(Z^*)$
  is $(1+\eps)^{\nu}$.  Let $avg$ be $(1+\eps)^{\nu}$ rounded up to
  the next value in $AVG$.  Then we know that
  $avg \leq (1+\eta_2) \cdot avg^*$.  The dynamic program also
  considers the state $Z=[J_{\nu-1}^*,b,avg]$.  However, the dynamic
  program computes the \emph{minimum} time point
  $t(Z_{\nu},Z) \geq t^{LB}$ so that the set of jobs
  $J\setminus J_{\nu -1}^*$ can be feasibly scheduled in
  $[0,t(Z_{\nu},Z))$ having \reservation\ cost not more than $b$.  This
  implies that $t(Z_{\nu},Z) \leq t(Z^*)$, which implies that
  $t(Z) \leq t(Z^*)$.  Note that $t^{LB}=0$ for the specified values
  in $Z$.

  Suppose, the statement is true for the iterations $1,2,\ldots,i-1$.
  We prove that it is also true for iteration $i$, in which we
  consider job sets from $\mathcal{F}_u$.  Again, let
  $Z_1^* = [J_{u+1}^*,b_1^*,avg_1^*]$ and
  $Z_2^* = [J_u^*,b_2^*,avg_2^*]$ with
  $J_{u+1}^* \in \mathcal{F}_{u+1}$ and $J_u^* \in \mathcal{F}_{u}$ be
  the states that represent $\mathcal{S}^*$.  By our hypothesis, we know that
  the dynamic program constructs a state $Z_1 =[J_{u+1},b_1,avg_1]$ with
  \begin{enumerate}[(i)]
  \item $J_{u+1} = J_{u+1}^*$,
  \item $b_1 \leq (1+\eta_1)^{i-1} \cdot b_1^*$,
  \item $avg_1 \leq (1+\eta_2)^{i-1}     \cdot avg_1^*$, and
  \item $t(Z_1) \leq t(Z_1^*)$.
  \end{enumerate}
  We augment this schedule in the following way.  Suppose, we \reserve\   the same slots that $\mathcal{S}^*$ \reserves\ for the jobs in
  $J_{u+1}^* \setminus J_{u}^*$ in the interval $[t(Z_1^*),t(Z^*_2))$.
  Let~$b_2$ be the resulting total \reservation\ cost after rounding up
  to the next value in $B$.  Thus, there is a feasible schedule for
  $J\setminus J_u^*$ having \reservation\ cost of at most
  \begin{eqnarray*}
    b_2 &\leq & (1+\eta_1)\cdot (b_1 + b_2^* - b_1^*) \\
    &\leq & (1+\eta_1)^{i} \cdot (b_1^* + b_2^* - b_1^*) \\
    &=& (1+\eta_1)^{i} \cdot b_2^*.
  \end{eqnarray*}
  The new average scheduling cost after rounding to the next value in
  $AVG$ is
  \begin{eqnarray*}
    avg_2 &\leq & (1+\eta_2) \cdot \frac{avg_1 \cdot t(Z_1) + (1+\eps)^{u+1}\cdot \left(t(Z_2^*) - t(Z_1)\right)}{t(Z_2^*)}\\
          &\leq & (1+\eta_2)^{i} \cdot \frac{avg^*_1 \cdot t(Z_1) + (1+\eps)^{u+1}\cdot \left(t(Z_2^*) - t(Z_1)\right)}{t(Z_2^*)}\\
          &\leq & (1+\eta_2)^{i} \cdot \frac{avg^*_1 \cdot t(Z^*_1) + (1+\eps)^{u+1}\cdot \left(t(Z_2^*) - t(Z^*_1)\right)}{t(Z_2^*)}\\
          &=& (1+\eta_2)^{i} \cdot avg_2^*.
  \end{eqnarray*}
  The third inequality follows from the fact that
  $avg_1^* \geq (1+\eps)^{u+1}$.  The dynamic program also considers
  the link between the state $Z_1$ and $Z_2:= [J_{u}^*,b_2,avg_2]$.
  We first observe that $t^{LB} \leq t(Z_2^*)$, since
  \[avg_2 \cdot t(Z_2^*) \geq avg_1 \cdot t(Z_1) + (1+\eps)^{u+1}\cdot
  (t(Z_2^*)-t(Z_1))\]
  by construction of $avg_2$.  Furthermore, we observe that
  $b_2 - b_1 \geq b_2^* - b_1^*$ by construction of $b_2$.  These two
  facts together with $t(Z_1)\leq t(Z_1^*)$ imply that
  $t(Z_1,Z_2) \leq t(Z_2^*)$, which implies that
  $t(Z_2)\leq t(Z_2^*)$.

To complete the proof, we need to specify the parameters $\eta_1$ and $\eta_2$.
We want that $(1+\eta_i)^{\nu} \leq (1+\eps)$ for $i=1,2$.
We claim that for a given $\nu \geq 1$ there exists an $\bar{\eta}>0$ such that for all $\eta \in (0,\bar{\eta}]$ we have $(1+\eta)^{\nu} \leq 1+2\nu\eta$.
Consider the function $f(\eta):=(1+\eta)^{\nu} - 1 - 2\nu\eta$.
We have that $f(0)=0$ and $f'(\eta) < 0$ for $\eta \in [0,2^{1/{(\nu-1)}}-1)$.
This shows the claim.
Hence, we choose $\eta_i = \min \{ \frac{\eps}{2\nu},2^{1/{(\nu-1)}}-1\}$ for $i=1,2$.
This shows the statement of the lemma and that the size of $B$ as well as the size of $AVG$ are bounded by a polynomial in the size of the input.
\end{proof}

We remark that the given DP works for more general \reservation\ cost
functions $e:\mathbb{N} \rightarrow \mathbb{Q}_{\geq 0}$ than
considered here in the paper.  As argued in the proof, it is sufficient for
the DP that there is a function $E(p,[t_1,t_2))$ that outputs in polynomial time for a given time
interval~$[t_1,t_2)$ and a given $p\in \mathbb{Z}_{\geq 0}$ the total cost of the~$p$ cheapest slots in
$[t_1,t_2)$.

We also remark that the running time of the presented DP is
exponential, because the size of the sets $\mathcal{F}_u$ are
exponential in the size of the input.  However, in the next section we
show that we can trim the sets $\mathcal{F}_u$ down to ones of
polynomial size at an arbitrarily small loss in the performance
guarantee.

\subsection{Trimming the state space}
\label{sec:trimming}

The set $\mathcal{F}_u$, containing all possible job sets $J_u$, is of
exponential size, and so is the DP state space. In the context of
scheduling with variable machine speed, it has been shown
in~\cite{MegowV2013} how to reduce the set $\mathcal{F}_u$ for a
similar DP (without \reservation\ decision, though) to a set
$\tilde{\mathcal{F}}_u$ of polynomial size at only a small loss in the
objective value. In general, such a procedure is not necessarily
applicable to our setting because of the different objective involving
additional \reservation\ cost and the different decision space. However,
the compactification in~\cite{MegowV2013} holds {\em independently of
  the speed of the machine} and, thus, independently of the
\reservation\ decision of the DP (interpret non/\reservation\ as speed
$0/1$). Hence, we can apply it to our cost-aware scheduling framework
and obtain a PTAS.  We now describe the building blocks for this
trimming procedure and argue why we can apply it in order to obtain
the set $\tilde{\mathcal{F}}_u$ for our problem.

\subsubsection{Light Jobs.}

The first building block for the trimming procedure is a
classification of the jobs based on their weights.

\begin{definition}
  Given a weight schedule and a job $j \in J$ with starting weight
  $S_j^w\in W\!I_u$, we call job $j$ \emph{light} if
  $w_j\le \eps^2 |W\!I_u|$, otherwise $j$ is called \emph{heavy}.
\end{definition}
This classification enables us to structure near-optimal solutions.
To impose structure on the set of light jobs, the authors in~\cite{MegowV2013} describe the following routine for a given weight schedule~$\mathcal{S}$.
First, remove all light jobs from $\mathcal{S}$ and move the remaining jobs within each interval $W\!I_u$ so that the idle weight in $W\!I_u$ is consecutive.
Then, schedule the light jobs according to the \emph{reverse Smith's rule}, that is, for each $u=1,\ldots,\nu$ and each idle weight $w\in W\!I_u$, process at $w$ a light job~$j$ that maximizes $p_{j}/w_{j}$.
Eventually, shift the processing of each interval $W\!I_u$ to $W\!I_{u+1}$, which delays the completion of every job by at most a factor of $(1+\eps)^2$.
This delay allows to completely process every light job in the weight interval where it starts processing.
It can be shown that the cost of the resulting schedule is at most a factor of $1+O\left(\eps\right)$ greater than the cost of $\mathcal{S}$, which brings us to the following structural statement.

\begin{lemma}[\!\!\cite{MegowV2013}]\label{lemma:LightJobs}
  At a loss of a factor of $1+O\left(\eps\right)$ in the
  scheduling cost, we can assume the following. For a given interval
  $W\!I_u$, consider any pair of light jobs $j,k$. If both jobs start
  in $W\!I_u$ or later and $p_k/w_k\le p_j/w_j$, then
  $C_j^w\le C_k^w$.
\end{lemma}

We remark, that Lemma~\ref{lemma:LightJobs} holds independently of the
speed of the machine, as pointed out in \cite{MegowV2013}. This means that at a loss of a factor of
$1+O\left(\eps\right)$ in the scheduling cost we can
assume also for our problem that light jobs are scheduled according to \textit{reverse
  Smith's rule} in the weight-dimension, which holds independently of our actual \reservation\ decision.

\subsubsection{Localization.}
We now localize jobs in the weight-dimension to gain more
structure. That is, we determine for every job $j\in J$ two values
$r_j^w$ and $d_j^w$ such that, independently of our actual \reservation\ decision, $j$ is scheduled completely within
$\left[r_j^w,d_j^w\right)$ in some
$(1+O\left( \eps\right))$-approximate weight-schedule (in
terms of the scheduling cost).  We call $r_j^w$ and $d_j^w$ the
\emph{release-weight} and the \emph{deadline-weight} of job $j$,
respectively.

\begin{lemma}[\!\!\cite{MegowV2013}]
  \label{lemma:Localization}
  We can compute in polynomial time values $r_j^w$ and $d_j^w$ for
  each $j\in J$ such that: (i) there exists a
  $\left(1+O\left(\eps\right)\right)$-approximate
  weight-schedule (in terms of the scheduling cost) that processes
  each job $j$ within $[r_j^w,d_j^w)$, (ii) there exists a constant
  $s\in O\left(\log\left(1/\eps\right)/\eps\right)$ such
  that $d_j^w\le r_j^w\cdot (1+\eps)^s$, (iii) $r_j^w$ and $d_j^w$ are
  integer powers of $(1+\eps)$, and (iv) the values $r_j^w$ an $d_j^w$
  are independent of the speed of the machine.
\end{lemma}

This lemma enables us to localize all jobs in $J$ in polynomial time
and independent of our actual \reservation\ decision, as guaranteed by
property (iv).

\subsubsection{Compact Search Space.} Based on the localization of
jobs in weight space, we can cut the number of different possibilities
for a candidate set $J_u$ in iteration~$u$ of our DP down to a
polynomial number.  That is, we replace the set $\mathcal{F}_u$ by a
polynomially sized set $\tilde{\mathcal{F}}_u$.  Instead of describing
all sets $S \in \tilde{\mathcal{F}}_u$ explicitly, the we give all
possible complements $R = J \setminus S$ and collect them in a set
$\mathcal{D}_u$, where a set $R\in \mathcal{D}_u$ represents a
possible set of jobs having completion weights in $W\!I_{u+1}$ or
later.  Obviously, a set $R \in \mathcal{D}_u$ must contain all jobs
$j \in J$ having a release weight $r_j^w \ge (1+\eps)^u$.
Furthermore, we know that $d_j^w \geq (1+\eps)^{u+1}$ is necessary for
job $j$ to be in a set $R \in \mathcal{D}_u$.  Following property (ii)
in Lemma~\ref{lemma:Localization}, we thus only need to decide about
the jobs having a release weight $r_j^w = (1+\eps)^i$ with
$i \in \left\{u+1-s,\ldots, u-1\right\}$.  An enumeration over
basically all possible job sets for each
$i \in \left\{u+1-s,\ldots, u-1\right\}$ gives the following desired
result.

\begin{lemma}[\!\!\cite{MegowV2013}]
  \label{lemma:CompactSearchSpace}
  For each $u$, we can construct in polynomial time a set
  $\tilde{\mathcal{F}}_u$ that satisfies the following: (i) there
  exists a $(1+O\left(\eps\right))$-approximate
  weight-schedule (in terms of the scheduling cost) in which the set
  of jobs with completion weight at most $(1+\eps)^u$ belongs to
  $\tilde{\mathcal{F}}_u$, (ii) the set $\tilde{\mathcal{F}}_u$ has
  cardinality at most
  $2^{O\left(\log^3(1/\eps)/\eps^2\right)}$, and (iii) the
  set $\tilde{\mathcal{F}}_u$ is completely independent of the speed
  of the machine.
\end{lemma}

Again, Property (iii) implies that we can construct the set
$\tilde{\mathcal{F}}_u$ independently of our \reservation\ decision.

To complete the proof of Theorem~\ref{thm:PTAS} it remains to argue on
the running time of the DP.  The DP has $\nu$ iterations, where in
each iteration for at most
$2^{O\left(\log^3(1/\eps)/\eps^2\right)} \cdot |B| \cdot
|AVG|$
previous states at most
$2^{O\left(\log^3(1/\eps)/\eps^2\right)} \cdot |B| \cdot
|AVG|$
many links to new states are considered.  Therefore, the running time
complexity of our DP is
$\nu \cdot (2^{O\left(\log^3(1/\eps)/\eps^2\right)} \cdot
|B| \cdot |AVG|)^2$,
which is bounded by a polynomial in the size of the input.

\section{Minimizing the makespan on unrelated machines}
\label{sec:Cmax}

Finally we derive positive results for the problem of minimizing makespan with \reservation\ costs on unrelated machines. The standard scheduling problem without \reservation\ cost
\abc{R}{pmtn}{C_{\max}} can be solved optimally in polynomial time by
solving a linear program as was shown by Lawler and
Labetoulle~\cite{LawlerL1978}.  We show that the problem complexity
does not increase significantly when taking into account time-varying
\reservation\ cost.

Consider the preemptive makespan minimization problem with \reservation\ cost.  Recall that we can use every machine in a \reserved\ time slot
and pay only once. Thus, it is sufficient to find an optimal
\reservation\ decision for solving this problem, because we can use the
polynomial-time algorithm in~\cite{LawlerL1978} to find the optimal
schedule within these slots.
\begin{observation}
  Given the set of time slots \reserved\ in an optimal solution, we can
  compute an optimal schedule in polynomial time.
\end{observation}

Given an instance of our problem, let $Z$ be the optimal makespan of
the relaxed problem {\em without} \reservation\ cost.  Notice that $Z$
is not necessarily integral.  To determine an optimal \reservation\ decision, we use the following observation.
\begin{observation}\label{obs:OptCmax}
  Given an optimal makespan $C_{\max}^*$ for
  \abc{R}{pmtn}{C_{\max}+E}, an optimal schedule \reserves\ the
  $\lceil Z \rceil$ cheapest slots before~$\lceil C_{\max}^* \rceil$.
\end{observation}

Note that we must pay full \tff\ for a used time slot, no
matter how much it is utilized, and so does an optimal solution.  In
particular, this holds for the last \reserved\ slot.  Hence, it remains
to compute  an optimal value $C^* :=\lceil C_{\max}^* \rceil$ which we
do by the following procedure.

We compute for every interval $I_k = \left[s_k,d_k\right)$,
$k=1,\ldots,K$, an optimal point in time for $C^*$ assuming that
$C^* \in I_k$. Hereby we restrict to relevant intervals $I_k$ which allow
for a feasible schedule, i.e., $s_k\geq \lceil Z \rceil$. For a
relevant interval $I_k$, we let $C^*=s_{k}$ and \reserve\ the $\lceil Z \rceil$ cheapest
time slots before $C^*$, which is optimal by
Observation~\ref{obs:OptCmax}. Notice that any \reserved\ time slot of cost~$e$ such that
$e > e_{k} +1$ can be replaced by a time slot from $I_{k}$ leading to
a solution of less total cost.  Thus, if there is no such time slot
then $s_{k}$ is the best choice for $C^*$ in $I_{k}$.
Suppose there is such a time slot that could be replaced. Let $R \subseteq \left\{1,\dots,k-1\right\}$ be the index
set of intervals that contain at least one \reserved\ slot.  We define
$I_{\ell}$ to be the interval with $e_{\ell} = \max_{h\in R} e_h$ and
denote by $r_{h}$ the number of \reserved\ time slots in $I_{h}$.
Replace $\min\{r_{\ell},d_k-s_k-r_k\}$ \reserved\ slots from $I_{\ell}$
by slots from $I_k$ and update $R$, $I_{\ell}$ and $r_k$.  This
continues until $e_{\ell} \leq e_{k}+1$ or the interval $I_k$ is
completely \reserved, i.e., $r_k = d_k-s_k$.  This operation takes at
most $O(K)$ computer operations per interval to compute the best
$C^*$-value in that interval.  It yields the following theorem.

\begin{theorem}\label{thm:SingleMachMakespan}
  The scheduling problem \abc{R}{pmtn}{C_{\max} + E} can be solved in
  polynomial time in the order of $O(K^2)$ plus the running time
  for solving \abc{R}{pmtn}{C_{\max}} without \reservation\ cost~\cite{LawlerL1978}.
\end{theorem}

\section{Conclusion}
We investigate basic scheduling problems within the framework
of time-varying costs or tariffs, where the processing of jobs
causes some time-dependent cost in addition to the usual QoS
measure.  We presented optimal algorithms and best possible
approximation algorithms for the scheduling objectives of minimizing
the makespan on unrelated machines and the sum of (weighted)
completion times on a single machine.

While our work closes the problems under consideration from an approximation point of view, it leaves open the approximability of multi-machine settings for the min-sum objective. Further research may also ask for the  complexity status when assuming that jobs have different release dates and for other natural objective functions such as average and maximum flow-time.

Our unrelated machine model is time-slot based, that is, a \reservation\ decision is made for a time slot and then all machines in this time slot are available. No less relevant appears to be the model with {\em machine-individual} \tffs, that is, a \reservation\ decision is made for a time slot on each machine
individually. It is not difficult to see that a standard LP can be adapted for optimally solving~\abc{R}{pmtn,r_j}{C_{\max}} with fractional \reservation\ cost. However, if time slots can be \reserved\
only integrally then the integrality gap for the simple LP is unbounded and the problems seems much harder.

Time-varying cost or tariffs appear in many applications in practice but they are hardly
 investigated from a theoretical perspective. With our work we settle the complexity status and approximability status for very classical scheduling problems. We hope to foster further research on this framework of time-varying costs or tariffs. We emphasize that the framework is clearly not restricted to cost-aware scheduling problems. Virtually any problem in which scarce resources are to be rented from some provider lends itself to be modelled in this way, with (vehicle) routing problems as a directly appealing example.

\bibliographystyle{abbrv} \bibliography{literature}
\end{document}